\documentclass{llncs}

\usepackage{amsmath}
\usepackage{amssymb}
\usepackage{subfigure}
\usepackage{pgfplots}
\usetikzlibrary{plotmarks}
\usepackage{multirow}
\usepackage{array}
\usepackage[ruled,vlined]{algorithm2e}
\usepackage{pgfplotstable,xstring}
\usepackage{url}
\definecolor{mygreen}{RGB}{38,160,0}

\usepackage{tikz}
\usetikzlibrary{automata,arrows,shapes,decorations.pathreplacing}
\usetikzlibrary{external}
\newcommand{\se}{\mathsf{sum\_err\_c}}

\newcommand{\seo}{\mathsf{sum\_err\_o}}
\newcommand{\solc}{{\sf solution}}
\newcommand{\sol}{{\sf solution}}

\renewcommand{\geq}{\geqslant}

\renewcommand{\leq}{\leqslant}

\begin{document}

\title{A Novel Combinatorial Method for Estimating Transcript Expression with RNA-Seq: \\Bounding the Number of Paths}

\author{
Alexandru I. Tomescu\inst{1} \and Anna Kuosmanen\inst{1} \and Romeo Rizzi\inst{2} \and Veli M{\"a}kinen\inst{1} 
}

\institute{Helsinki Institute for Information Technology HIIT, \\
Department of Computer Science, University of Helsinki, Finland\\
\email{\{tomescu,aekuosma,vmakinen\}@cs.helsinki.fi}
\and
Department of Computer Science, University of Verona, Italy\\
\email{romeo.rizzi@univr.it}
}

\maketitle

\begin{abstract}

RNA-Seq technology offers new high-throughput ways for transcript identification and quantification based on short reads, and has recently attracted great interest. The problem is usually modeled by a weighted splicing graph whose nodes stand for exons and whose edges stand for split alignments to the exons. The task consists of finding a number of paths, together with their expression levels, which optimally explain the coverages of the graph under various fitness functions, such least sum of squares. In (Tomescu \emph{et al.} RECOMB-seq 2013) we showed that under general fitness functions, if we allow a polynomially bounded number of paths in an optimal solution, this problem can be solved in polynomial time by a reduction to a min-cost flow program. In this paper we further refine this problem by asking for a bounded number $k$ of paths that optimally explain the splicing graph. This problem becomes NP-hard in the strong sense, but we give a fast combinatorial algorithm based on dynamic programming for it.  In order to obtain a practical tool, we implement three optimizations and heuristics, which achieve better performance on real data, and similar or better performance on simulated data, than state-of-the-art tools Cufflinks, IsoLasso and SLIDE. Our tool, called Traph, is available at \url{http://www.cs.helsinki.fi/gsa/traph/}

\end{abstract}

\section{Introduction}\label{sect:introduction}

In this paper we tackle a biological multi-assembly problem \cite{citeulike:177150} motivated by the recent RNA-Seq technology \cite{Mortazavi2008MammalRNASeq,pwm09,citeulike:8493482}: reconstruct as accurately as possible the RNA transcripts of a gene, given only a set of short RNA reads sequenced from them. The transcripts are concatenations of exons, the difficulty of the problem arising from the fact that they can have some identical exons.

\begin{figure}[!tpb]
\centerline{
\subfigure[\label{fig:problem-example-a}]{\scalebox{0.6}{
\begin{tikzpicture}[scale=.7, -angle 45]
\node[label={[label distance=-1mm]above left:$8$},circle,draw=black,inner sep = -2mm,minimum size = 4mm,fill=white] (a) at (0.5,0) {$a$};
\node[label={[label distance=-0.5mm]90:$6$},circle,draw=black,inner sep = -2mm,minimum size = 4mm,fill=white] (b) at (2.5,1) {$b$};
\node[label={[label distance=0mm]-90:$3$},circle,draw=black,inner sep = -2mm,minimum size = 4mm,fill=white] (e) at (2.5,-1) {$e$};
\node[label={[label distance=-0.5mm]90:$5$},circle,draw=black,inner sep = -2mm,minimum size = 4mm,fill=white] (c) at (6,1) {$c$};
\node[label={[label distance=0mm]-90:$3$},circle,draw=black,inner sep = -2mm,minimum size = 4mm,fill=white] (f) at (6,-1) {$f$};
\node[label={[label distance=-1mm]above right:$8$},circle,draw=black,inner sep = -2mm,minimum size = 4mm,fill=white] (d) at (8,0) {$d$};
\draw (a) to node[draw=none,fill=none,auto] {$5$} (b);
\draw (a) to node[draw=none,fill=none,below] {$3$} (e);
\draw (e) to node[draw=none,fill=none,auto] {$3$} (b);
\draw (b) to node[draw=none,fill=none,auto] {$4$} (f);
\draw (b) to node[draw=none,fill=none,auto] {$5$} (c);
\draw (c) to node[draw=none,fill=none,auto] {$5$} (d);
\draw (e) to node[draw=none,fill=none,below] {$3$} (f);
\draw (f) to node[draw=none,fill=none,below] {$4$} (d);
\end{tikzpicture}
}}\hspace{.15cm}\subfigure[\label{fig:problem-example-b}]{\scalebox{0.6}{
\begin{tikzpicture}[scale=.7, -angle 45]
\node[label={[label distance=-1mm]above left:{\textcolor{red}{$5$}$+$\textcolor{blue}{$3$}}},circle,draw=black,inner sep = -2mm,minimum size = 4mm,fill=white] (a) at (0.5,0) {$a$};
\node[label={[label distance=-0.5mm]90:{\textcolor{red}{$5$}}},circle,draw=black,inner sep = -2mm,minimum size = 4mm,fill=white] (b) at (2.5,1) {$b$};
\node[label={[label distance=0mm]-90:\textcolor{blue}{$3$}},circle,draw=black,inner sep = -2mm,minimum size = 4mm,fill=white] (e) at (2.5,-1) {$e$};
\node[label={[label distance=-0.5mm]90:{\textcolor{red}{$5$}}},circle,draw=black,inner sep = -2mm,minimum size = 4mm,fill=white] (c) at (6,1) {$c$};
\node[label={[label distance=0mm]-90:\textcolor{blue}{$3$}},circle,draw=black,inner sep = -2mm,minimum size = 4mm,fill=white] (f) at (6,-1) {$f$};
\node[label={[label distance=-1mm]above right:{\textcolor{red}{$5$}$+$\textcolor{blue}{$3$}}},circle,draw=black,inner sep = -2mm,minimum size = 4mm,fill=white] (d) at (8,0) {$d$};
\draw[line width=1.5pt,draw=red] (a) to node[draw=none,fill=none,auto] {\textcolor{red}{$5$}} (b);
\draw[line width=1.5pt,draw=blue] (a) to node[draw=none,fill=none,below] {\textcolor{blue}{$3$}} (e);
\draw (e) to node[draw=none,fill=none,auto] {$3$} (b);
\draw (b) to node[draw=none,fill=none,auto] {$4$} (f);
\draw[line width=1.5pt,draw=red] (b) to node[draw=none,fill=none,auto] {\textcolor{red}{$5$}} (c);
\draw[line width=1.5pt,draw=red] (c) to node[draw=none,fill=none,auto] {\textcolor{red}{$5$}} (d);
\draw[line width=1.5pt,draw=blue] (e) to node[draw=none,fill=none,below] {\textcolor{blue}{$3$}} (f);
\draw[line width=1.5pt,draw=blue] (f) to node[draw=none,fill=none,below] {\textcolor{blue}{$3$}} (d);
\end{tikzpicture}
}}\hspace{.15cm}\subfigure[\label{fig:problem-example-c}]{\scalebox{0.6}{
\begin{tikzpicture}[scale=.7, -angle 45]
\node[label={[label distance=-1mm]above left:\textcolor{red}{$5$}$+$\textcolor{blue}{$3$}},circle,draw=black,inner sep = -2mm,minimum size = 4mm,fill=white] (a) at (0.5,0) {$a$};
\node[label={[label distance=-0.5mm]90:\textcolor{red}{$5$}$+$\textcolor{blue}{$3$}},circle,draw=black,inner sep = -2mm,minimum size = 4mm,fill=white] (b) at (2.5,1) {$b$};
\node[label={[label distance=0mm]-90:\textcolor{blue}{$3$}},circle,draw=black,inner sep = -2mm,minimum size = 4mm,fill=white] (e) at (2.5,-1) {$e$};
\node[label={[label distance=-.5mm]90:\textcolor{red}{$5$}},circle,draw=black,inner sep = -2mm,minimum size = 4mm,fill=white] (c) at (6,1) {$c$};
\node[label={[label distance=0mm]-90:\textcolor{blue}{$3$}},circle,draw=black,inner sep = -2mm,minimum size = 4mm,fill=white] (f) at (6,-1) {$f$};
\node[label={[label distance=-1mm]above right:\textcolor{red}{$5$}$+$\textcolor{blue}{$3$}},circle,draw=black,inner sep = -2mm,minimum size = 4mm,fill=white] (d) at (8,0) {$d$};
\draw[line width=1.5pt,draw=red] (a) to node[draw=none,fill=none,auto] {\textcolor{red}{$5$}} (b);
\draw[line width=1.5pt,draw=blue] (a) to node[draw=none,fill=none,below] {\textcolor{blue}{$3$}} (e);
\draw[line width=1.5pt,draw=blue] (e) to node[draw=none,fill=none,auto] {\textcolor{blue}{$3$}} (b);
\draw[line width=1.5pt,draw=blue] (b) to node[draw=none,fill=none,auto] {\textcolor{blue}{$3$}} (f);
\draw[line width=1.5pt,draw=red] (b) to node[draw=none,fill=none,auto] {\textcolor{red}{$5$}} (c);
\draw[line width=1.5pt,draw=red] (c) to node[draw=none,fill=none,auto] {\textcolor{red}{$5$}} (d);
\draw (e) to node[draw=none,fill=none,below] {$3$} (f);
\draw[line width=1.5pt,draw=blue] (f) to node[draw=none,fill=none,below] {\textcolor{blue}{$3$}} (d);
\end{tikzpicture}
}}}

\caption{An example for $k = 2$, and fitness function $f_v(x) = x^2$, $f_{uv}(x) = x^2$, for all nodes $v$, and edges $(u,v)$\label{fig:problem-example}. In Fig.~\ref{fig:problem-example-a}, a splicing directed acyclic graph; its nodes and edges are labeled with their observed average coverage. In Fig.~\ref{fig:problem-example-b}, the optimal 2 paths for Problem~$2$-\textsf{UTEO}, with expression levels 5 and 3; their cost is $1 + 1 = 2$ (from node $b$, and edge ($f$,$d$), respectively). In the case of Problem~$2$-\textsf{UTEC}, we have to add $3^2 + 4^2$ to their cost (from uncovered edges ($e$,$b$), ($b$,$f$)), which is not optimal. In Fig.~\ref{fig:problem-example-c}, the optimal 2 paths for Problem~$2$-\textsf{UTEC}, with expression levels 5 and 3; their cost is $2^2 + (1 + 1 + 3^2) = 15$ (from node $b$, and edges ($b$,$f$), ($f$,$d$), ($e$,$f$), respectively).}
\end{figure}
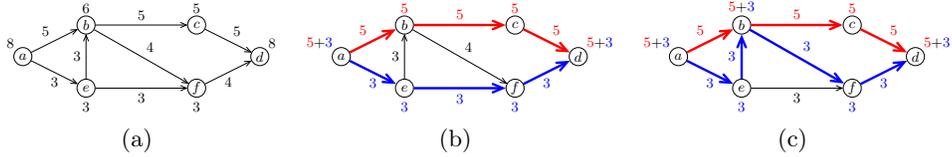

This problem has attracted great interest from the community, resulting in tools such as Cufflinks \cite{Tra10}, IsoInfer/IsoLasso \cite{DBLP:conf/recomb/FengLJ10,LFJ11}, SLIDE \cite{Lin:2012uq}, CLIIQ \cite{Li:2011fk}, Scripture \cite{Guttman2010}, iReckon \cite{iReckon}, TRIP \cite{TRIP}, NSMAP \cite{citeulike:9303305}, Montebello \cite{citeulike:12016403}, FlipFlop \cite{bernard:hal-00803134}. The methods rely on a graph model, the most common being a \emph{splicing graph} \cite{Heber01072002}. Its nodes represent contiguous stretches of DNA uninterrupted by spliced reads (called \emph{pseudo-exons}), while its edges are derived from overlaps, or from spliced read alignments. The splicing graph is directed and acyclic (a DAG), the orientation of the edges being according to the starting position of the pseudo-exons inside the genome. Every node $v$ has an associated observed average coverage, computed as the total number of reads aligned to the pseudo-exon $v$, divided by the exon length. Similarly, every edge $(u,v)$ has an associated coverage, which is the total number of reads split aligned to the junction between pseudo-exons $u$ and $v$. 

The biological problem translates to covering the graph with intersecting paths, under different cost models, such as least sum of squares (IsoInfer/IsoLasso, SLIDE), least sum of absolute differences (CLIIQ). Many of the above mentioned tools work by exhaustively enumerating all possible (combinations of) paths, with some restrictions, and then estimating their fitness with an Integer Linear Program,  Quadratic Program, or a QP + LASSO regression. Cufflinks computes a path cover with a minimum number of paths, and only in a second step estimates their expression levels. 

In \cite{Tomescu:2013fk} we introduced a novel very general framework, encompassing many of the previous proposals; according to the survey \cite{rna-seq-survey}, it can be classified as \emph{de novo} genome-based, since it does not use annotation information. Apparently, parallel to our work, a similar min-cost flow approach, called FlipFlop, was proposed in \cite{bernard:hal-00803134}. Our method assumes that for every node $v$ and edge $(u,v)$ of the splicing graph, we are given fitness functions $f_v$ and $f_{uv}$ which penalize the difference between the observed average coverage and the predicted coverage. The problem was translated as finding (an unlimited number of) paths with associated expression levels such that the sum of all penalties is minimum. For example, if for every node or edge $z$, $f_z(x) = x^2$, then we have a least sum of squares model as in IsoInfer/IsoLasso and SLIDE, and if $f_z(x) = x$ we have a least sum of absolute difference model as in CLIIQ (see Fig.~\ref{fig:problem-example} for an example).Various other fitting functions can be considered, such as $f_z(x) = x/cov(z)$ \cite{STAN:STAN232}, or $f_v(x) = x^2 * length(v)^2$. Therein, we proposed a min-cost flow method to solve this problem in polynomial-time, assuming the fitting functions are convex, which was competitive with Cufflinks and IsoLasso. In that approach, the size of the solution is polynomially bounded, but it was left open to find even more parsimonious optimal or good solutions.

We now tackle the problem of optimally covering the splicing graph with a bounded number $k$ of paths (see Fig.~\ref{fig:problem-example} for an example). This is relevant in practice since a small fraction of the graph can be erroneous due to various biological events or technical errors, like template switching, self-priming, reading errors, wrong splicing alignment \cite{citeulike:399605,citeulike:8493482,citeulike:11526711,citeulike:9386912}. 

\begin{problem}[$k$-\textsf{UTEC}]
Given a splicing DAG $G=(V,E)$ with positive coverage values $cov(v)$ and $cov(u,v)$, integer $k \geq 1$, and fitting functions $f_v(\cdot)$ and $f_{uv}(\cdot)$, for all $v \in V$ and $(u,v) \in E$, the \emph{$k$-Unannotated Transcript Expression Cover} Problem is to find a tuple $\mathcal{P}$ of $k$ paths from the sources of $G$ to the sinks of $G$,
with an estimated expression level $e(P)$ for each path $P \in \mathcal{P}$, which minimize 
\begin{align*}
\se(G,\mathcal{P}) := &\sum_{v\in V} f_v\Big(\Big|cov(v) - \sum_{P \in \mathcal{P}\text{ s.t. } v\in P} e(P) \Big|\Big) + &\\
&\hfill \sum_{(u,v)\in E} f_{uv}\Big(\Big|cov(u,v) - \sum_{P \in \mathcal{P}\text{ s.t. } (u,v)\in P} e(P) \Big|\Big).&
\end{align*}
\label{problem:utec}
\end{problem}
We also study the following outlier sensitive variant asking for $k$ paths which best fit to the coverage only of the nodes and edges that they contain. 
\begin{problem}[$k$-\textsf{UTEO}]
Under the same assumptions as for Problem~$k$-\textsf{UTEC}, the \emph{$k$-Unannotated Transcript Expression Outlier} Problem is to find a tuple $\mathcal{P}$ of $k \geq 1$ paths from the sources of $G$ to the sinks of $G$, with an estimated expression level $e(P)$ for each path $P \in \mathcal{P}$, which minimize
\begin{align*}
\seo(G,\mathcal{P}) := & \sum_{\begin{subarray}{c}P \in \mathcal{P},\, v\in P\end{subarray}} f_v\Big(\Big|cov(v) - \sum_{P \in \mathcal{P}\text{ s.t. } v\in P} e(P) \Big|\Big) + &\\
&\sum_{\begin{subarray}{c}P \in \mathcal{P},\, (u,v)\in P\end{subarray}} f_{uv}\Big(\Big|cov(u,v) - \sum_{P \in \mathcal{P}\text{ s.t. } (u,v)\in P} e(P) \Big|\Big).&
\end{align*}
\end{problem}

In Sec.~\ref{sec:np-hardness-proof} we show that both problems are NP-hard in the strong sense.\footnote{We should note that a preliminary version of this paper was presented as a poster at the RECOMB-seq, April 2013, conference \url{http://bioinfo.au.tsinghua.edu.cn/recomb2013/upload/programseq.pdf}, and that our NP-hardness proof has already inspired the NP-hardness proof \cite{maximumlikelihood} of the isoform reconstruction by maximum likelihood problem, deployed in tools such as iReckon, NSMAP, Montebello.} Nevertheless, in Sec.~\ref{sec:dynamic-programming} we give dynamic programming algorithms with a time-complexity of $O(|M|^k(n^2+\Delta^k)n^{k})$, where $M$ is the set of all possible expression levels, and $\Delta$ is the maximum in-degree of the graph. To obtain a practical implementation of these algorithms, we apply, as explained in Sec.~\ref{sec:implementation}, the following optimizations and heuristics: 
\begin{enumerate}
\item We decompose the problem along cut nodes, i.e., we find a node whose removal disconnects the graph into two components, and recursively solve the problem on the two subgraphs.

\item We employ a genetic algorithm for finding the optimal expression levels: the fitness of a given $k$-tuple of expression levels is the cost of the optimal paths having these expression levels, obtained by our dynamic programming in time $O((n^2+\Delta^k)n^{k})$; experiments show that the genetic algorithm has very small variability in practice. 

\item In order to reduce the exponential dependency on $k$, we choose a $k' \leq k$ depending on the size of the graph and guaranteeing that the problem is tractable, then compute the optimal $k'$ paths, remove their weight from the graph, and recurse until obtaining $k$ paths in total. 

\end{enumerate}

Experimental results, given in Sec.~\ref{sec:experimental-results}, show that our algorithm, together with the above optimizations and heuristics, has better performance on real RNA-Seq data, and similar or better performance on simulated data, than our min-cost flow method, and than state-of-the-art tools Cufflinks, IsoLasso, and SLIDE. In these experiments, we run the program for all values of $k$ up to a bound depending on the size of the input graph, and choose the $k$ such that the paths returned by the program have the minimum value of the objective function. However, the choice of $k$ is highly customizable by the user.

\section{Methods}

\subsection{The NP-Hardness Proof}
\label{sec:np-hardness-proof}


\begin{theorem}
\label{thm:bp-hard}
If the cost functions $f_v$ and $f_{uv}$ are such that $f_v(0) = 0$, $f_{uv}(0) = 0$, and $f_v(x) > 0$, $f_{uv}(x) > 0$ for all $x >0$ and all nodes $v$ and edges $(u,v)$ of the input splicing graph, then Problems~$k$-\textsf{UTEC} and $k$-\textsf{UTEO} are NP-hard in the strong sense.
\end{theorem}

\begin{proof}
We follow the proof of \cite{Vatinlen20081390}, underlining the differences in what follows. We reduce from \textsf{3-PARTITION}. In this problem, we are given a set $A = \{a_1,\dots,a_{3q}\}$ with $3q$ elements, and for all $a \in A$, a positive integer $s(a)$, its size, such that $B/4 < s(a) < B/2$ and $\sum_{a \in A}s(a) = qB$. We are asked whether there exists a partition of $A$ into $q$ disjoint sets each of size $B$.

Given an instance $(A,s)$ to \textsf{3-PARTITION}, we construct (see also Fig.~\ref{fig:reduction}) the graph $G_{A,s}$ having:
\begin{itemize}
\item $V(G_{A,s}) = \{s,x_1,\dots,x_{3q},y,z_1,\dots,z_q,t\}$,
\item for every $i \in \{1,\dots,3q\}$, we add arcs $(s,x_i)$, $(x_i,y)$ to $G_{A,s}$, both with coverage $s(a_i)$, and also set the coverage of $x_i$ to $s(a_i)$,
\item for every $i \in \{1,\dots,q\}$, we add arcs $(y,z_i)$ and $(z_i,t)$ to $G_{A,s}$, both with coverage $B$, and also set the coverage of $z_i$ to $B$,
\item the coverage of $s$, $y$ and $t$ is $qB$.
\end{itemize}

\begin{figure}[t]
\centering
\subfigure[\label{fig:reduction}]{\scalebox{0.7}{
\begin{tikzpicture}[scale=.7, -angle 45]
\node[label={[label distance=-1mm]above left:$qB$},circle,draw=black,inner sep = -2mm,minimum size = 6mm,fill=white] (s) at (0,0) {$s$};
\node[label={[label distance=-1mm]above:$qB$},circle,draw=black,inner sep = -2mm,minimum size = 6mm,fill=white] (y) at (5,0) {$y$};
\node[label={[label distance=-1mm]80:$a_1$},circle,draw=black,inner sep = -2mm,minimum size = 6mm,fill=white] (x1) at (2.5,2) {$x_1$};
\node[label={[label distance=-1mm]-80:$a_2$},circle,draw=black,inner sep = -2mm,minimum size = 6mm,fill=white] (x2) at (2.5,.5) {$x_2$};
\node[label={[label distance=-1mm]-80:$a_{3q}$},circle,draw=black,inner sep = -2mm,minimum size = 6mm,fill=white] (xn) at (2.5,-2) {$x_{3q}$};
\draw (s) to node[draw=none,fill=none,auto] {$a_1$} (x1);
\draw (s) to node[draw=none,fill=none,below] {$a_2$} (x2);
\draw (s) to node[draw=none,fill=none,below left] {$a_{3q}$} (xn);
\draw (x1) to node[draw=none,fill=none,auto] {$a_1$} (y);
\draw (x2) to node[draw=none,fill=none,below] {$a_2$} (y);
\draw (xn) to node[draw=none,fill=none,below right] {$a_{3q}$} (y);
\node[draw=none,fill=none] (dots) at (2.5,-.6) {$\vdots$};

\node[label={[label distance=-1mm]80:$B$},circle,draw=black,inner sep = -2mm,minimum size = 6mm,fill=white] (z1) at (7.5,2) {$z_1$};
\node[label={[label distance=-1mm]-80:$B$},circle,draw=black,inner sep = -2mm,minimum size = 6mm,fill=white] (z2) at (7.5,.5) {$z_2$};
\node[label={[label distance=-1mm]-80:$B$},circle,draw=black,inner sep = -2mm,minimum size = 6mm,fill=white] (zn) at (7.5,-2) {$z_{q}$};
\node[label={[label distance=-1mm]above right:$qB$},circle,draw=black,inner sep = -2mm,minimum size = 6mm,fill=white] (t) at (10,0) {$t$};
\draw (y) to node[draw=none,fill=none,auto] {$B$} (z1);
\draw (y) to node[draw=none,fill=none,below] {$B$} (z2);
\draw (y) to node[draw=none,fill=none,below left] {$B$} (zn);
\draw (z1) to node[draw=none,fill=none,auto] {$B$} (t);
\draw (z2) to node[draw=none,fill=none,below] {$B$} (t);
\draw (zn) to node[draw=none,fill=none,below right] {$B$} (t);
\node[draw=none,fill=none] (dots) at (7.5,-.6) {$\vdots$};
\end{tikzpicture}
}}\hspace{1cm}
\subfigure[\label{fig:dynamic-programming}]{\includegraphics[width=3.8cm]{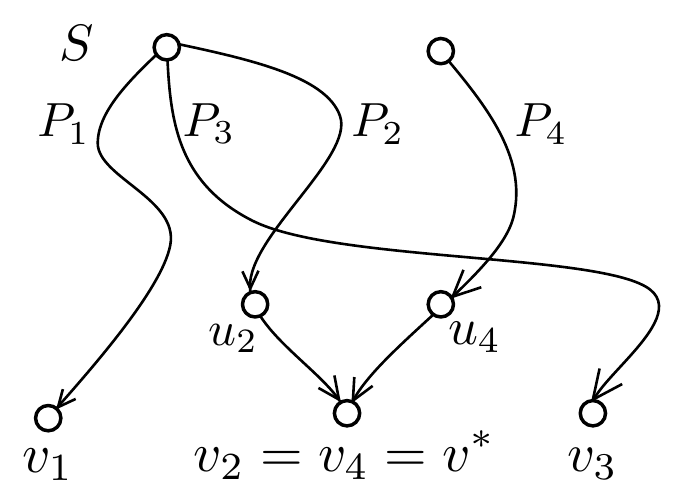}}
\caption{In Fig.~\ref{fig:reduction}, a reduction of \textsf{3-PARTITION} to Problems~$k$-\textsf{UTEC} or $k$-\textsf{UTEO}. In Fig.~\ref{fig:dynamic-programming}, computing $\sol(v_1,v_2,v_3,v_4)$. We assume that $v_2 = v_4$ is a sink of $G_{v_1,\dots,v_k}$, and it is chosen as $v^\ast$; we then enumerate through all pairs of vertices from $N^-(v^\ast) \times N^-(v^\ast)$; in this case, we find $(u_2,u_4)$ and we extend the optimal paths ending in $(v_1,u_2,v_3,u_4)$ with the edges $(u_2,v^\ast)$ and $(u_4,v^\ast)$}
\end{figure}
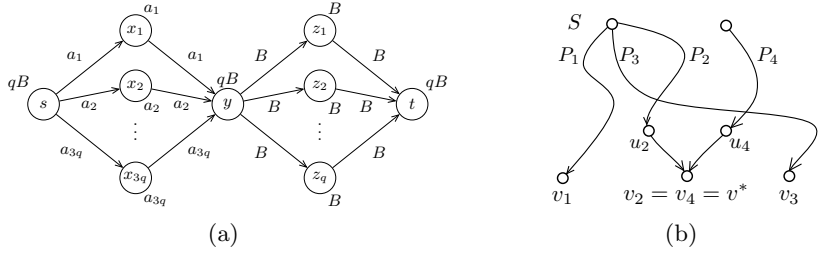

We prove that there exists a partition of $A$ into $q$ sets of size $B$ if and only if Problem~$k$-\textsf{UTEC} admits on $G_{A,s}$ a solution with cost 0 made up of at most $3q$ paths, and analogously for Problem~$k$-\textsf{UTEO}.

For the forward implication, let $A_1,\dots,A_q$ be a partition of $A$ into $q$ sets of size $B$. To obtain a solution to Problem~$k$-\textsf{UTEC} with cost 0, for every $A_i = \{a_{i_1},a_{i_2},a_{i_3}\}$ we add to the solution the three paths $(s,x_{i_1},y,z_i,t)$, $(s,x_{i_2},y,z_i,t)$, $(s,x_{i_3},y,z_i,t)$. These three paths completely cover the edges $(s,x_{i_1}), (x_{i_1},y)$, $(s,x_{i_2}), (x_{i_2},y)$, and $(s,x_{i_3}), (x_{i_3},y)$, respectively, and they are the only paths to do so, since $A_1,\dots,A_q$ is a partition of $A$. This results in a zero cost to be added to the objective function. Moreover, since $s(a_{i_1}) + s(a_{i_2}) + s(a_{i_3}) = B$, then these three paths together completely cover the edges $(y,z_i)$ and $(z_i,t)$. This also implies a zero cost to be added to the objective function.

For the backward implication, observe that a solution to Problem~$k$-\textsf{UTEC} with cost 0 and at most $3q$ paths must have exactly $3q$ paths. To see why this is the case, observe that the sum of the expression levels of the paths is exactly $qB$, as they pass through node $y$, and this is a zero-cost solution. Moreover, the sum of the coverages of vertices $x_1,\dots,x_{3q}$ is $qB$, by construction. The fact that this is a zero-cost solution thus implies that each of the $3q$ vertices $x_1,\dots,x_{3q}$ must be covered by at least one of the $3q$ paths. Therefore, each $x_i$ is covered by exactly one path. 

For every $i \in \{1,\dots,q\}$, let $Q_i$ denote the set of paths in this optimal solution covering node $z_i$. As this is a zero-cost solution, the sum of their expression levels is $B$, and their expression levels belong to $A$. Since $B/4 < s(a) < B/2$, for all $a \in A$, then each $Q_i$ contains exactly three paths. This entails that for any $1 \leq i < j \leq q$, $Q_i \cap Q_j = \emptyset$. Therefore, by associating with each $i \in \{1,\dots,q\}$ the subsets of $A$ that correspond to the first arc of the three paths of $Q_i$ we obtain a partition of $A$ into $q$ sets of size $B$.

The proof of \cite[Proposition 2]{Vatinlen20081390} can be followed identically from this point onwards to show that this is a pseudo-polynomial reduction. The proof for Problem~$k$-\textsf{UTEO} is entirely similar.
\qed\end{proof}

\subsection{The Dynamic Programming Algorithms}
\label{sec:dynamic-programming}

Onwards, we propose dynamic programming algorithms for Problems~$k$-\textsf{UTEO} and $k$-\textsf{UTEC}. Since the algorithm for Problem~$k$-\textsf{UTEO} is simpler than for Problem~$k$-\textsf{UTEC}, we present the former here, and defer the latter to the full version of this paper. 

We will assume that the possible expression levels of the paths in an optimal solution belong to a finite set $M$. Our strategy is to find the optimal $k$-tuple of paths having a fixed $k$-tuple of expression levels. The solution is then obtained by enumerating all $k$-tuples of expression levels from $M^k$, and taking the $k$-tuple of paths having the smallest value of the objective function. Despite the dependency on the expression levels' values, having the two steps separated means that we can employ, for a practical implementation, any local search heuristic for finding the optimal expression levels. This search will be guided by the cost of the objective function returned by the dynamic programming; the search can be done at any chosen granularity, eventually including a priori information about the true expression levels. We employ a genetic algorithm which behaves well in practice (see Sec.~\ref{sec:implementation}).

Let us assume from now on one such choice $(e_1,\dots,e_k)$ of expression levels fixed. The main difficulty behind the algorithm is that the paths can share vertices. Accordingly, we have to process all $k$-tuples of vertices of $V$; for every $(v_1,\dots,v_k) \in V^k$, we define 
\begin{align*}
&\sol(v_1,\dots,v_k) := \min_{\begin{subarray}{c}\text{paths $P_1,\dots,P_k$ in $G$,} \\ \text{each $P_i$ is from a source to $v_i$} \end{subarray}}\hspace{-.6cm} \seo(G,P_1,\dots,P_k).&
\end{align*}

Since the input directed graph is acyclic, we let $\prec$ be a topological order on $V$, and define the partial order $\prec^k$ on $V^k$ as follows:
$$(v_1',\dots,v_k') \prec^k (v_1,\dots,v_k) \text{\ \ \ iff\ \ \ } \exists i \in \{1,\dots,k\} \text{ such that } v_i' \prec v_i.$$

\begin{sloppypar}
\noindent Then, the computation of $\sol$ is done by dynamic programming, by enumerating the tuples $(v_1,\dots,v_k) \in V^k$ in the order $\prec^k$, and computing $\sol(v_1,\dots,v_k)$ from the previous values, according to $\prec^k$, as indicated in Algorithm~\ref{alg:se-recurrence}, where $N^-(v)$ denotes the in-neighborhood of a node $v$, and $G_{v_1,\dots,v_k}$ denotes the subgraph of $G$ induced by the vertices from which there is a directed path to one of $v_1,\dots,v_k$ (see Fig.~\ref{fig:dynamic-programming} for a sketch). 
\end{sloppypar}

%
%


\begin{algorithm}[t]
\caption{Computing $\sol(v_1,\dots,v_k)$ for a fixed tuple $(e_1,\dots,e_k)$ of expression levels, for Problem~$k$-\textsf{UTEO}.\label{alg:se-recurrence}}

\SetKwBlock{sol}{${\sf solution}(v_1,\dots,v_k)$}{end}

\tcc{\scriptsize initialization for all possible source tuples}
\ForEach{$(s_1,\dots,s_k) \in S^k$}{
$\displaystyle {\sf solution}(s_1,\dots,s_k) \leftarrow \sum_{\begin{subarray}{c}i \in \{1,\dots,k\}\text{ such that}\\\forall i' < i,\, s_{i} \neq {s_{i'}}\end{subarray}}f_{s_i}\Big(\Big|cov(s_i) - \sum_{\begin{subarray}{c}j \in \{i,\dots,k\}\\\text{s.t. }s_{j} = s_i\end{subarray}}e_j\Big|\Big)$\;
}

\sol{
	$min \leftarrow \infty$\;
	let $v^\ast$ be a sink of $G_{v_1,\dots,v_k}$, which is not the source of $G$\;
	let $i_1,\dots,i_\ell$ be all the positions in the tuple $(v_1,\dots,v_k)$ where $v^\ast$ appears\;
	\tcc{\scriptsize we enumerate through all $\ell$-tuples of in-neighbors of $v^\ast$}
	\ForEach{$(u_{i_1},\dots,u_{i_\ell}) \in N^-(v^\ast)^\ell$}{
			\tcc{\scriptsize we get the optimal cost for such a tuple}
			$err \leftarrow {\sf solution}(v_1,\dots,v_{i_1-1},u_{i_1},v_{i_1+1},v_{i_\ell-1},u_{i_\ell},v_{i_\ell+1},\dots,v_k)$\;
			\tcc{\scriptsize we sum up the cost of covering $v^\ast$ with the $\ell$ paths extended from $u_{i_1},\dots,u_{i_\ell}$ having expression levels $e_{i_1},\dots,e_{i_\ell}$}
			$\displaystyle err \leftarrow err + f_{v^\ast}\Big(\Big|c(v^\ast) - \sum_{j=1}^\ell e_{i_j}\Big|\Big)$\;
			\tcc{\scriptsize we sum up the cost of covering the edges $(u_{i_1},v^\ast), \dots, (u_{i_\ell},v^\ast)$ with the $\ell$ paths extended from $u_{i_1},\dots,u_{i_\ell}$ having expression levels $e_{i_1},\dots,e_{i_\ell}$, respectively}
			$\displaystyle err \leftarrow err + \hspace{-.5cm}\sum_{\begin{subarray}{c}j \in \{1,\dots,\ell\}\text{ such that}\\\forall j' < j,\; u_{i_{j}} \neq u_{i_{j'}}\end{subarray}}\hspace{-.5cm}f_{u_{i_j}v^\ast}\Big(\Big|c(u_{i_j},v^\ast) - \sum_{\begin{subarray}{c}t \in \{j,\dots,\ell\}\\\text{s.t. }u_{i_t} = u_{i_j}\end{subarray}} e_{i_t}\Big|\Big)$\;
			\textbf{if} {$err < min$} \textbf{then} $min \leftarrow err$\;
	}
	{\bf return} $min$.
}
\end{algorithm}

\begin{theorem}
If the cost functions $f_v$ and $f_{uv}$ are such that $f_v(x) \geq 0$, $f_{uv}(x) \geq 0$ for all $x \geq 0$ and all nodes $v$ and edges $(u,v)$ of the input splicing graph, then Problems~$k$-\textsf{UTEO} and $k$-\textsf{UTEC} can be solved in time $O(|M|^k(n^2+\Delta^k)n^{k})$, where $n := |V(G)|$, we assume that $M$ is the set of possible expression levels, and the maximum in-degree of $G$ is $\Delta$.
\label{thm:algorithm-cover}
\end{theorem}

\begin{proof} We give the proof only for Problem~$k$-\textsf{UTEO}. The algorithm and the proof for Problem~$k$-\textsf{UTEC} are analogous, but more involved; they will be presented in the full version of this paper.

Let $(P_1,\dots,P_k)$ be a tuple of $k$ \emph{optimal} paths starting in a source and ending in $v_1,\dots,v_k$, i.e.,
\begin{align}
&\seo(G,P_1,\dots,P_k) =\hspace{-1.05cm}\min_{\begin{subarray}{c}\indent\vspace{.05cm}\\ \text{paths $Q_1,\dots,Q_k$ in $G$,} \\ \text{each $Q_i$ is from a source to $v_i$} \end{subarray}}\hspace{-1.05cm} \seo(G,Q_1,\dots,Q_k).&
\label{proof:eqn0}
\end{align}

\noindent Let $v^\ast$ be a sink of $G_{v_1,\dots,v_k}$ which is not a source of $G$ (if none such node exists, then all $v_1, \dots, v_k$, are sources and the value of $\sol(v_1,\dots,v_k)$ has already been set). Assume also that $i_1,\dots,i_\ell$, $\ell \geq 1$, are the positions in the tuple $(v_1,\dots,v_k)$ where $v^\ast$ appears (see Fig.~\ref{fig:dynamic-programming} for a sketch).

Let $u_{i_1},\dots,u_{i_\ell}$ be the predecessors of $v^\ast$ on $P_{i_1},\dots,P_{i_\ell}$, respectively. For every $j \in \{1,\dots,\ell\}$, denote by $P^\ast_{i_j}$ the path $P_{i_j}$ from which we remove its last node, $v^\ast$. To simplify notation, denote by $(P_1,\dots,P^\ast,\dots,P_k)$ the tuple $(P_1,\dots,P_k)$ in which, for every $j \in \{1,\dots,\ell\}$, we replace $P_{i_j}$ by $P_{i_j}^\ast$. Similarly, we denote by $(v_1,\dots,u,\dots,v_k)$ the tuple $(v_1,\dots,v_k)$ in which, for every $j \in \{1,\dots,\ell\}$, we replace $v_{i_j}$ by $u_{i_j}$.

From the fact that $v^\ast$ is a sink of $G_{v_1,\dots,v_k}$, neither the node $v^\ast$, nor the edges in $new\_edges := \{(u_{i_1},v^\ast)$,\dots,$(u_{i_\ell},v^\ast)\}$, belong to any path in $G$ ending in $\{v_1,\dots,v_k\} \setminus \{v^\ast\}$. Therefore, the following relation holds:

\begin{equation}
\begin{split}
&\seo(G,P_1,\dots,P_k) = \seo(G,P_1,\dots,P^\ast,\dots,P_k) + \\
& + f_{v^\ast}\Big(\Big|c(v^\ast) - \sum_{j=1}^\ell e_{i_j}\Big|\Big) +
\hspace{-.5cm}\sum_{\begin{subarray}{c}j \in \{1,\dots,\ell\}\text{ such that}\\\forall j' < j,\; u_{i_{j}} \neq u_{i_{j'}}\end{subarray}}\hspace{-.5cm}f_{u_{i_j}v^\ast}\Big(\Big|c(u_{i_j},v^\ast) - \sum_{\begin{subarray}{c}t \in \{j,\dots,\ell\}\\\text{s.t. }u_{i_t} = u_{i_j}\end{subarray}} e_{i_t}\Big|\Big).
\end{split}
\label{proof:eqn1}
\end{equation}

\begin{sloppypar}
Let $(P_1',\dots,P_k')$ be any tuple of $k$ paths from a source to $(v_1,\dots,u,\dots,v_k)$ such that $\seo(G,(P_1',\dots,P_k')) = \sol(v_1,\dots,u,\dots,v_k)$. From the fact that also the paths $P_1,\dots,P^\ast,\dots,P_k$ end in $v_1,\dots,u,\dots,v_k$, respectively, and the optimality of $(P_1',\dots,P_k')$, we have
\end{sloppypar}
\begin{equation}
\seo(G,P_1,\dots,P^\ast,\dots,P_k) \geq \seo(G,P_1',\dots,P_k').
\label{proof:eqn2}
\end{equation}

From (\ref{proof:eqn2}) and (\ref{proof:eqn1}) we get
\begin{equation}
\begin{split}
&\seo(G,P_1,\dots,P_k) \geq \seo(G,P_1',\dots,P_k') + \\
& + f_{v^\ast}\Big(\Big|c(v^\ast) - \sum_{j=1}^\ell e_{i_j}\Big|\Big) + 
\hspace{-.5cm}\sum_{\begin{subarray}{c}j \in \{1,\dots,\ell\}\text{ such that}\\\forall j' < j,\; u_{i_{j}} \neq u_{i_{j'}}\end{subarray}}\hspace{-.5cm}f_{u_{i_j}v^\ast}\Big(\Big|c(u_{i_j},v^\ast) - \sum_{\begin{subarray}{c}t \in \{j,\dots,\ell\}\\\text{s.t. }u_{i_t} = u_{i_j}\end{subarray}} e_{i_t}\Big|\Big).
\end{split} 
\label{proof:eqn3}
\end{equation}

Let us denote by $(P_1',\dots,P'\cup\{v\},\dots,P_k)$ the tuple $(P_1',\dots,P_k')$ in which, for every $j \in \{1,\dots,\ell\}$, we add node $v^\ast$ at the end of path $P_{i_j}$. Resorting again to the fact that the node $v^\ast$ and the edges in $new\_edges$ do not belong to any path in $G$ ending in $\{v_1,\dots,v_k\} \setminus \{v^\ast\}$, we get that the right-hand side of the inequality (\ref{proof:eqn3}) is equal to $\seo(G,P_1',\dots,P'\cup\{v^\ast\},\dots,P_k')$, thus:
\begin{align}
&\seo(G,P_1,\dots,P_k) \geq \seo(G,P_1',\dots,P'\cup\{v^\ast\},\dots,P_k').&
\label{proof:eqn4}
\end{align}

To conclude, if we enumerate through all $(v_{i_1}',\dots,v'_{i_\ell}) \in N^-(v^\ast)^\ell$, we will find the nodes $u_{i_1},\dots,u_{i_\ell}$ preceding $v^\ast$ on the optimal paths $P_{i_1},\dots,P_{i_\ell}$, respectively. If we extend each optimal path ending in $v_1,\dots,u,\dots,v_k$ by the node $v^\ast$, we obtain paths $P_1',\dots,P'\cup\{v^\ast\},\dots,P_k'$ whose cost is optimal, by (\ref{proof:eqn0}) and (\ref{proof:eqn4}).

\begin{sloppypar}
Once table $\sol$ is computed, the solution for Problem~\textsf{$k$-UTEO} is $\min_{(t_1,\dots,t_k) \in T^k} \solc(t_1,\dots,t_k)$,
\noindent where $T$ is the set of sinks of $G$. Finally, note that if we store in $\sol(v_1,\dots,v_k)$ also the predecessors $u_1,\dots,u_k$ of $v_1,\dots,v_k$ on the optimal paths from the sources, we can then trace back the $k$ optimal paths from sources to sinks.
\end{sloppypar}

The time complexity bound follows from the fact that there are $n^k$ tuples $(v_1,\dots,v_k) \in V^k$, computing the sinks of $G_{v_1,\dots,v_k}$ takes time linear in the number of edges of $G_{v_1,\dots,v_k}$, which are $O(n^2)$, and there are at most $\Delta^k$ candidate predecessors $u_{i_1},\dots,u_{i_\ell}$ of $v^\ast$ on $P_{i_1},\dots,P_{i_\ell}$, respectively, in the in-neighborhood of $v^\ast$.
\qed
\end{proof}


\subsection{Optimizations and heuristics for a practical implementation}
\label{sec:implementation}
We implemented the two algorithms in our tool Traph \cite{Tomescu:2013fk}, which can be seen as an alternative version of our min-cost flow method.

In order to speed-up the dynamic programming, we look for nodes whose removal disconnects the input graph $G$. If $v$ is such a cut node, then we consider the subgraphs $G_1$, induced by $v$ and the nodes from which there is a directed path to $v$, and $G_2$, induced by $v$ and the nodes to which there is a directed path starting from $v$. It is clear that $v$ is the only sink in $G_1$ and the only source in $G_2$. We can then recursively solve Problem \textsf{$k$-UTEC}/\textsf{$k$-UTEO} for $G_1$, and use the optimal solution for $v$ in $G_1$ as initialization to the dynamic programming table of $G_2$, and solve the problem on $G_2$. 

Moreover, in order to avoid enumerating through all tuples of possible expression levels, at present we employ a genetic algorithm: for each individual $(e_1,\dots,e_k)$ in the current population, its fitting function $\phi(e_1,\dots,e_k)$ is the cost of the optimal paths having the expression levels $(e_1,\dots,e_k)$, computed with our dynamic programming algorithms. 

As by its nature genetic algorithm results vary, Traph iterates a hundred times each run, and the algorithm is rerun five times with different initial values. The solution with the smallest value of the objective function is then chosen.
We tested the variability of the results by running Traph a hundred times on a sample containing alignments from one gene. All hundred runs resulted to the same transcripts, with the standard deviation of the path weight being less than 0.001\% of the mean path weight for every transcript.

In order to reduce the exponential dependency on $k$, we choose a number $k' \leq k$ which depends on $n$ and $\Delta$ such that $O((n^2+ \Delta^k)n^k)$ is small for practical purposes. We then solve the problem by looking for the best $k'$ paths with their optimal expression levels, remove their weights from the graph, and then recurse until obtaining $k$ paths as requested.

We use a least sum of squares model, i.e., the fitness function that we use is $f_z(x) = x^2$, for all nodes and edges $z$. In the present work, we also tried the fitting function $f_v(x) = x^2 * length(v)^2$ (which tries to explain the total coverage of the exons, not only their average coverage), and in \cite{Tomescu:2013fk} we tried the fitting function $f_v(x) = x/cov(v)$, but they did not give better results.

\section{Experimental results}
\label{sec:experimental-results}

We compared Traph cover (Problem~\textsf{UTEC}), and Traph outlier (Problem~\textsf{UTEO}) with Cufflinks~\cite{Tra10}, IsoLasso~\cite{LFJ11}, SLIDE~\cite{Li:2011fk}, and to our min-cost flow method \cite{Tomescu:2013fk}. We tried to compare also against Scripture, iReckon, CLIIQ, but we ran into compatibility issues in installing them, or we could not get reliable results. Even though Traph is not yet employing paired-end read information,  the experiments (both simulated and real) were conducted with paired-end reads, and Cufflinks, IsoLasso and SLIDE had access to the paired-end information. As Traph is a \emph{de novo} genome-based tool, we ran the other tools without annotation. Our experiment setup and validation criteria are the same as in \cite{Tomescu:2013fk}.\footnote{In \cite{Tomescu:2013fk} we use \emph{bitscore} instead of sequence dissimilarity, which is based on normalized compression distance and is better grounded as a measure. However, it needs full alignment to be output. Here we approximate this measure with sequence dissimilarity, which is fast to compute using Myers's bitparallel algorithm \cite{DBLP:journals/jacm/Myers99}, and this enables much larger data sets to be evaluated within the time constraints of article submission.} We run the algorithms for all values of $k$ up to a bound depending on the size of the input graph, and choose the minimum $k$ giving the minimum of the objective function. However, the choice of $k$ can be easily customized by the user.

\smallskip
\noindent \emph{Matching criteria.} In order to match the predicted transcripts with the true transcripts, we take into account sequences but also expression levels. For each gene, we construct a bipartite graph with the true transcripts $\mathcal{T} = (T_1,T_2,\dots)$ as nodes in one set of the bipartition, and the predicted transcripts $\mathcal{P} = (P_1, P_2,\dots)$ as nodes in the other set of the bipartition. Empty sequences with $0$ expression level were added so that both sets of the bipartition had an equal number of nodes. The cost of an edge between a true transcript $T_i$ with expression level $e(T_i)$ and a predicted transcript $P_j$ with expression level $e(P_j)$ is defined as a combined measure between: (i) the edit distance between $T_i$ and $P_j$, divided by $\max(|T_i|,|P_j|)$, which we call \emph{sequence dissimilarity}; and (ii) the ratio $|e(P_j) - e(T_i)|/e(T_i)$, which we call \emph{relative expression level difference} (see \cite{Tomescu:2013fk} for further details). The minimum weight perfect matching was then computed; this gives a one-to-one mapping between true and predicted transcripts. Each matched node pair with relative expression level difference and sequence dissimilarity under given thresholds define a true positive event (TP). The other kind of nodes define false positive (FP) and false negative (FN) events, depending on side of the bipartite graph they reside. The prediction efficiency is based on precision = TP/(TP+FP), recall = TP/(TP+FN) and F-measure = 2$*$precision$*$recall/(precision+recall).

\newcommand{\plota}{
	\begin{tikzpicture}	
	\pgfplotsset{ 
		tick label style={font=\scriptsize}, 
		label style={font=\scriptsize}, 
		legend style={font=\tiny,draw=none,fill=none},
		legend pos = {north west},
		every axis x label/.append style={
			at={(0.5,0)}, 
			below,
			yshift = 5pt
		},
		every axis y label/.append style={
			at={(0,0.5)}, 
			below,
			yshift = -3pt
		}
	}

\begin{axis}[height=4cm, width=6cm, ylabel={F-measure}, xlabel={sequence dissimilarity}, enlargelimits=false, 
	xmin=0.08,xmax=0.92,
    ymin=0,ymax=.3,
    xtick={.1,.3,.5,.7,.9},
    xticklabels={10\%,30\%,50\%,70\%,90\%},
    ytick={0,.1,.2,.3},
    y tick label style={
    /pgf/number format/.cd,
    fixed,
    precision=3
  }
]

	\pgfsetplotmarksize{2.5pt}	
	\addplot[mark=o,color=black] table[x=bitscore,y=fscore] {plots/wabi/singles/Cufflinks (Average over all genes) ed0.1.stat};\addlegendentry{Cufflinks};
	\addplot[mark=square,color=blue] table[x=bitscore,y=fscore] {plots/wabi/singles/Isolasso (Average over all genes) ed0.1.stat};\addlegendentry{IsoLasso};
	\addplot[mark=+,color=orange] table[x=bitscore,y=fscore] {plots/wabi/singles/Slide (Average over all genes) ed0.1.stat};\addlegendentry{SLIDE};
	\addplot[mark=triangle,color=red] table[x=bitscore,y=fscore] {plots/wabi/singles/Traph flow (Average over all genes) ed0.1.stat};\addlegendentry{Min-cost flow};
	\addplot[mark=star,color=mygreen] table[x=bitscore,y=fscore] {plots/wabi/singles/Traph dynamic cover (Average over all genes) ed0.1.stat};\addlegendentry{Traph cover};
	\end{axis}

	\end{tikzpicture}
}
\newcommand{\plotb}{
	\begin{tikzpicture}	
	\pgfplotsset{ 
		tick label style={font=\scriptsize}, 
		label style={font=\scriptsize}, 
		legend style={font=\tiny,draw=none,fill=none},
		legend pos = {south east},
		every axis x label/.append style={
			at={(0.5,0)}, 
			below,
			yshift = 5pt
		},
		every axis y label/.append style={
			at={(0,0.5)}, 
			below,
			yshift = -3pt
		}
	}

\begin{axis}[height=4cm, width=6cm, xlabel={sequence dissimilarity}, enlargelimits=false, 
	xmin=0.08,xmax=0.92,
    ymin=0,ymax=.3,
    xtick={.1,.3,.5,.7,.9},
    xticklabels={10\%,30\%,50\%,70\%,90\%},
    ytick={0,.1,.2,.3,.4,.5,.6,.8},
]

	\pgfsetplotmarksize{2.5pt}	
	\addplot[mark=o,color=black] table[x=bitscore,y=fscore] {plots/wabi/singles/Cufflinks (Average over all genes) ed0.4.stat};
	\addplot[mark=square,color=blue] table[x=bitscore,y=fscore] {plots/wabi/singles/Isolasso (Average over all genes) ed0.4.stat};
	\addplot[mark=+,color=orange] table[x=bitscore,y=fscore] {plots/wabi/singles/Slide (Average over all genes) ed0.4.stat};
	\addplot[mark=triangle,color=red] table[x=bitscore,y=fscore] {plots/wabi/singles/Traph flow (Average over all genes) ed0.4.stat};
	\addplot[mark=star,color=mygreen] table[x=bitscore,y=fscore] {plots/wabi/singles/Traph dynamic cover (Average over all genes) ed0.4.stat};
	\end{axis}

	\end{tikzpicture}
}

\newcommand{\plotc}{
	\begin{tikzpicture}	
	\pgfplotsset{ 
		tick label style={font=\scriptsize}, 
		label style={font=\scriptsize}, 
		legend style={font=\tiny,draw=none,fill=none},
		legend pos = {north east},
		every axis x label/.append style={
			at={(0.5,0)}, 
			below,
			yshift = 5pt
		},
		every axis y label/.append style={
			at={(0,0.5)}, 
			below,
			yshift = -3pt
		}
	}

\begin{axis}[height=4cm, width=6cm, xlabel={sequence dissimilarity}, enlargelimits=false, 
	xmin=0.08,xmax=0.92,
    ymin=0,ymax=0.07,
    xtick={.1,.3,.5,.7,.9},
    xticklabels={10\%,30\%,50\%,70\%,90\%},
    ytick={0,.035,.07},
    y tick label style={
    /pgf/number format/.cd,
    fixed,
    precision=3
  },
]

	\pgfsetplotmarksize{2.5pt}	
	\addplot[mark=o,color=black] table[x=bitscore,y=fscore] {plots/wabi/batch/cufflinks ed0.1.txt};
	\addplot[mark=square,color=blue] table[x=bitscore,y=fscore] {plots/wabi/batch/isolasso ed0.1.txt};
	\addplot[mark=+,color=orange] table[x=bitscore,y=fscore] {plots/wabi/batch/slide ed0.1.txt};
	\addplot[mark=triangle,color=red] table[x=bitscore,y=fscore] {plots/wabi/batch/traph flow ed0.1.txt};
	\addplot[mark=star,color=mygreen] table[x=bitscore,y=fscore] {plots/wabi/batch/traph dynamic cover ed0.1.txt};
	\end{axis}

	\end{tikzpicture}
}

\newcommand{\plotd}{
	\begin{tikzpicture}	
	\pgfplotsset{ 
		tick label style={font=\scriptsize}, 
		label style={font=\scriptsize}, 
		legend style={font=\tiny,draw=none,fill=none},
		legend pos = {south east},
		every axis x label/.append style={
			at={(0.5,0)}, 
			below,
			yshift = 5pt
		},
		every axis y label/.append style={
			at={(0,0.5)}, 
			below,
			yshift = -3pt
		}
	}
\begin{axis}[height=4cm, width=6cm, xlabel={sequence dissimilarity}, enlargelimits=false, 
	xmin=0.08,xmax=0.92,
    ymin=0,ymax=.3,
    xtick={.1,.3,.5,.7,.9},
    ytick={0,.1,.2,.3,.4,.6,.8},
    xticklabels={10\%,30\%,50\%,70\%,90\%}
    ]

	\pgfsetplotmarksize{2.5pt}	
	\addplot[mark=o,color=black] table[x=bitscore,y=fscore] {plots/wabi/batch/cufflinks ed0.4.txt};
	\addplot[mark=square,color=blue] table[x=bitscore,y=fscore] {plots/wabi/batch/isolasso ed0.4.txt};
	\addplot[mark=+,color=orange] table[x=bitscore,y=fscore] {plots/wabi/batch/slide ed0.4.txt};
	\addplot[mark=triangle,color=red] table[x=bitscore,y=fscore] {plots/wabi/batch/traph flow ed0.4.txt};
	\addplot[mark=star,color=mygreen] table[x=bitscore,y=fscore] {plots/wabi/batch/traph dynamic cover ed0.4.txt};
	\end{axis}

	\end{tikzpicture}
}

\newcommand{\plottable}{
\scriptsize
\begin{tabular}{|l|c|c|c|c|c|c|}\hline 
\multirow{3}{*}{Tool} & Total & \multicolumn{5}{c|}{Shared with annotation at} \\
& predicted & \multicolumn{5}{c|}{sequence dissimilarity under} \\
&  & \ 10\% & \ 20\%\ &\ 30\%\ &\ 40\%\ &\ 50\%\ \\\hline 
Cufflinks & 1916 & 648 & 955 & 1171 & 1307 & 1413 \\\hline 
IsoLasso & 1468 & 589 & 782 & 923 & 1022 & 1100 \\\hline 
SLIDE & 2229 & 635 & 983 & 1242 & 1391 & 1474 \\\hline 
Min-cost flow & 2148 & 722 & 1000 & 1228 & 1341 & 1456 \\\hline 
Traph cover & 2109 & \textbf{788} & \textbf{1063} & \textbf{1283} & \textbf{1407} & \textbf{1501} \\\hline 
\end{tabular}
}

\begin{figure}[t]
\subfigure[\texttt{s.} expr. diff. $\leq10\%$\label{fig:plot-c}]{\plota}%
\subfigure[\texttt{s.} expr. diff. $\leq 40\%$\label{fig:plot-d}]{\plotb}%
\\
\subfigure[\texttt{b.} expr. diff. $\leq 10\%$\label{fig:plot-e}]{\hspace{.3cm}\plotc}%
\subfigure[\texttt{b.} expr. diff. $\leq40\%$\label{fig:plot-f}]{\plotd}%
\\
\begin{center}
\subfigure[The total number of transcripts reported by the tools on real data\label{table:real-data}]{\raisebox{10mm}\plottable}%
\end{center}
%
%
%
\vspace{-.3cm}
\caption{Performance of the tools on simulated and real data. Plots \ref{fig:plot-c} and \ref{fig:plot-d} depict results in the \texttt{singles} scenario, and plots \ref{fig:plot-e} and \ref{fig:plot-f} depict results in the \texttt{batch} scenario; the legend for all plots is as in Fig.~\ref{fig:plot-c}. Real data results are in Fig.~\ref{table:real-data}.\label{fig:main-results}}
\vspace{-.3cm}
\end{figure}

\subsection{Simulated human data} 

For creating the simulated data we used the annotated genes in human chromosome 2 as reported by Ensembl database. Excluding the genes that had no transcripts as long or longer than the fragment size, we were left with 1,462 genes. We simulated reads with the RNASeqReadSimulator\footnote{http://www.cs.ucr.edu/\textasciitilde liw/rnaseqreadsimulator.html} by first choosing an expression level for each transcript at random from lognormal distribution with mean $-4$ and variance $1$, and then creating paired-end reads with fragment length mean 300 and standard deviation 20, with the starting positions of the fragments being chosen uniformly inside the transcripts. As argued in the case of IsoLasso~\cite{LFJ11}, various error models can be incorporated in these steps, but we chose to compare the performance of the methods in neutral conditions.

We devised two experiment setups. In the first one, which we call \texttt{singles}, 300,000 paired-end reads were generated independently from the transcripts of each of the genes, with the already assigned expression levels. They were independently given to TopHat \cite{TPS09} for alignment, and these independent alignment results were fed to each tool. In the second, more realistic experiment, which we call \texttt{batch}, we randomly chose 100 of the genes, chose expression levels for them with the same distribution as before and simulated 100 $*$ 300,000 reads as above. All these reads were fed to TopHat for alignment, and these combined alignment results were fed to the tools. The fragment length mean and standard deviation were passed to the tools, except for Cufflinks in \texttt{batch}, when it was able to infer them automatically; as our simulated data did not contain any single exon genes, SLIDE was unable to infer the fragment distribution, and it was given the fragment length mean and the standard deviation.

Cufflinks, IsoLasso and SLIDE were ran with the default parameters, because the parameters they offer relate to RNA-seq lab protocol, which was not simulated; we could not see changes to other parameters which could be relevant to the prediction. SLIDE's results are highly dependent on the lambda values, and as such it encourages the user to manually adjust the lambda values if the result set seems to either be missing isoforms or contain too many short isoforms, but for the sake of having automated tests we used the lambda values SLIDE estimated from the data as is. We use FPKM values as expression levels. Full simulated experiment input data is available on the webpage of Traph.

Fig.~\ref{fig:main-results} shows selected validation results. For each experiment, we choose two thresholds for the relative expression level differences, namely $10\%$ and $40\%$. Overall, Traph cover outperformed Traph outlier; at the moment we consider Traph cover as the default implementation, and plan to apply Traph outlier to other multi-assembly problems. In the \texttt{singles} scenario, Cufflinks, our min-cost flow method and Traph cover have very similar F-measure and out-perform IsoLasso and SLIDE. In the \texttt{batch} scenario, we obtain the same situation when the relative expression level difference is allowed to be at most $10\%$, but the min-cost flow method and Traph cover out-perform the other tools when the threshold for relative expression level difference is $40\%$, Traph cover giving slightly better results. Note that in the \texttt{batch} scenario the tools predicted transcripts which fall outside the annotated gene areas, which we accounted as FP events in the plots. For the 100 genes, Cufflinks predicted 512 transcripts inside gene areas, and 215 outside gene areas, Isolasso had 384 predictions inside and a surprising 7,422 outside; SLIDE had 725 inside and 94 outside; Traph cover had 458 inside and 98 outside; the min-cost flow method had 413 inside and 74 outside.

\smallskip
\noindent \emph{Running times.} On the batch dataset of reads from 100 genes, Cufflinks ran in 421 min, IsoLasso in 38 min, SLIDE in 257 min. Our script for creating the splicing graphs is written in Python and took 180 min; the min-cost flow method ran on these splicing graphs in 117 min, and Traph cover in 538 min.


\subsection{Real human data} We used the same real dataset from the IsoLasso paper, Caltech RNA-Seq track from the ENCODE project [GenBank:SRR065504], consisting of 75bp paired-end reads. Out of these reads, we picked the 2,406,339 which mapped to human chromosome~2. We selected the 735 genes where all tools made some prediction; these genes have 6,325 annotated transcripts. 

The transcripts predicted by each tool are matched with the annotated transcripts, employing the same minimum weight perfect matching method introduced before, but without taking into account expression levels. A true positive is a match selected by the perfect matching with varying sequence dissimilarity threshold. We present these results in Table~\ref{table:real-data}, where we note that Traph cover reports the most transcripts which match the annotation at all thresholds of sequence dissimilarity.

\section{Conclusion}
In this paper we tackled two multi-assembly problems arising from transcript identification and quantification with RNA-Seq, which ask for the $k$ paths which best explain, under given fitting functions, the coverages of a splicing graph. In our experiments we worked with least sum of squares  as fitting function, but our method supports very general fitting functions. We expect that these two models and algorithms to be applicable to other multi-assembly problems, such as in metagenomics or in viral quasi-species assembly. 

The two problems considered, Problem~$k$-\textsf{UTEO} and $k$-\textsf{UTEC}, are shown to be NP-hard in the strong sense, proof which already inspired a similar NP-hardness proof \cite{maximumlikelihood} of another problem pertaining to multi-assembly. If some of the input parameters are bounded ($k$, the maximum in-degree of the graph, the set of possible expression levels), then the problems can be solved in polynomial-time using dynamic programming. Nonetheless, in order to obtain a practical implementation, we considered three optimizations and heuristics which work well in practice, and in a feasible amount of time: on real data we report more annotated transcripts than Cufflinks, IsoLasso, SLIDE and our previous min-cost flow method, while on simulated data we obtain similar or better performance.


\section*{Acknowledgements}

We wish to thank Antti Honkela for many insightful discussions on transcript prediction. We also thank Travis Gagie for discussions on the NP-hardness proof. This work was partially supported by Academy of Finland under grant 250345 (CoECGR).

\bibliographystyle{splncs_srt}
\bibliography{paper_short}

\end{document}